\documentclass[a4paper]{article}
\usepackage{fullpage,times}
\usepackage{amsmath,amsfonts,amssymb}
\usepackage{amsthm}
\usepackage{algorithm}
\usepackage[noend]{algorithmic} 
\usepackage{multirow}
\usepackage{hyperref}

\newtheorem{theorem}{Theorem}
\newtheorem{proposition}{Proposition}
\newtheorem{lemma}{Lemma}
\newtheorem{corollary}{Corollary}

\newcommand{\N}{\mathcal{N}}
\newcommand{\K}{\mathcal{K}}

\newcommand{\E}{\mathcal{E}}
\newcommand{\bv}{\mathbf{v}}

\newcommand{\hv}{\hat{v}}

\newcommand{\csa}{FNE$_{sa}$}
\newcommand{\caa}{FNE$_{aa}$}
\newcommand{\B}{$\mathcal{B}$}

\begin{document}

\title{\bf Towards Better Models of Externalities in Sponsored Search Auctions
}
\date{}
%
%
%
%
%

%
\author{
%
%
Nicola Gatti\\
       Politecnico di Milano\\
       Milano, Italy\\
       nicola.gatti@polimi.it
\and
Marco Rocco\\
       Politecnico di Milano\\
       Milano, Italy\\
       marco.rocco@polimi.it
\and 
Paolo Serafino\\
       Teesside University\\
       Middlesbrough, UK\\
       p.serafino@tees.ac.uk
\and  
Carmine Ventre\\
       Teesside University\\
       Middlesbrough, UK\\
       c.ventre@tees.ac.uk
}

\maketitle

\begin{abstract}
Sponsored Search Auctions (SSAs) arguably represent the problem at the intersection of computer science and economics with the deepest applications in real life. Within the realm of SSAs, the study of the effects that showing one ad has on the other ads, a.k.a. externalities in economics, is of utmost importance and has so far attracted the attention of much research. However, even the basic question of modeling the problem has so far escaped a definitive answer. The popular cascade model is arguably too idealized to really describe the phenomenon yet it allows a good comprehension of the problem. Other models, instead, describe the setting more adequately but are too complex to permit a satisfactory theoretical analysis. In this work, we attempt to get the best of both approaches: firstly, we define a number of general mathematical formulations for the problem in the attempt to have a rich description of externalities in SSAs and, secondly, prove a host of results drawing a nearly complete picture about the computational complexity of the problem. We complement these approximability results with some considerations about mechanism design in our context.
\end{abstract}
\section{Introduction}
The computation of solutions maximizing the \emph{social welfare}, i.e., maximizing the total ``happiness'' of the advertisers, in sponsored search auctions (SSAs) strongly depends on how such happiness is defined. Clearly, the more clicks their ads receive, the more content advertisers are. A naive measure to forecast clicks, named click through rate (CTR), would only consider the \emph{quality} of the ad itself (``better'' ads receive more clicks). However, one cannot overlook the importance of \emph{externalities} in this context: specifically, \emph{slot-dependent externalities} (i.e., ads positioned higher in the list have a higher chance to get a click) and \emph{ad-dependent externalities} (e.g., the ad of a strong competitor -- e.g., BMW -- shown in the first slot can only decrease the number of  clicks that the ad -- e.g., of Mercedes -- in the second slot gets).  Much research focused on modeling externalities in SSAs and providing algorithms for the resulting optimization problem. 

On one hand of the scale, there is the simple, yet neat, \emph{cascade} model \cite{cascade,DBLP:conf/wine/AggarwalFMP08}. In the cascade model, users are assumed to scan the ads \emph{sequentially} from top to bottom and the probability with which a user clicks on the ad $a_i$ shown in slot $s_m$ is the product of the intrinsic quality $q_i$ of the ad, the relevance $\lambda_m$ of slot $s_m$ ({slot-dependant externality}) and of \emph{all} the ads allocated to slots $s_1$ through $s_{m-1}$. A host of results is proved in this model as the input parameters vary (e.g., $\lambda_m \in \{0,1\}$ rather than $\lambda_m \in [0,1]$). In its more general version, the optimization problem of social welfare maximization is conjectured to be NP-hard, shown to be in APX (i.e., a $1/4$-approximation algorithm is given) and shown to admit a QPTAS (a quasi-polynomial time approximation scheme) \cite{cascade}. In addition to its unknown computational complexity, the cascade model has two main limitations to be considered a satisfactory model of externalities in SSAs. First, it assumes that users have unlimited ``memory'' and that, consequently, an ad in slot $s_1$ exerts externalities to an ad many slots below. This is experimentally disproved in \cite{Segal} wherein it is observed how the \emph{distance} between ads is important. Second, it assumes that the externality of an ad is the same no matter what ad is exerted on. Nevertheless, while BMW can have a strong externality on Mercedes since both makers attract the high end of market, the externality on makers in a different price bracket, e.g., KIA, is arguably much less strong.

On the other hand of the scale, we can find models 
that try to address these limitations. 
In \cite{Fotakis} Fotakis \emph{et al.} propose a model whereby users have limited memory, i.e., externalities occur only within a \emph{window} of $c$ consecutive slots, and consider the possibility that externalities boost CTRs (\emph{positive} externalities) as well as reduce CTRs (\emph{negative} externalities).
In particular, the externalities of an ad apply to ads displayed $c$ slots below (\emph{forward} externalities) and ads displayed $c$ slots above (\emph{backward} externalities).
Moreover, in order to model the fact that externalities might have ad-dependent effect, they introduce the concept of \emph{contextual graph}, whereby vertices represent ads and edge weights represent the externality between the endpoints.
Their model turned out to be too rich to allow tight and significant algorithmic results 
(their main complexity results apply to the arguably less interesting case of forward positive externalities).

\subsection{Our contribution} 
The present work can be placed in the middle of this imaginary spectrum of models for externalities in SSAs. Our main aim is to enrich the literature by means of more general ways to model slot- and ad-dependent externalities, while giving a (nearly) complete picture of the computational complexity of the problem. We do not attempt to explicitly model the user's behavior but 
bridge the aforementioned models in order to overcome the respective weaknesses.
In detail, we enrich the naive model of SSAs by adding the concepts of window 
and contextual externalities, while keeping ad- and slot-dependent externalities factorized as in the cascade model.
We also complement much of the known literature by studying a model wherein the externalities coming from ads and slots cannot be expressed as a product.
Our study gives rise to a number of novel and rich models for which we can provide (often tight) approximability results (see Table \ref{tab:results} for an overview).\footnote{It is important to notice that, as common in the literature on SSAs, the number of slots is a parameter of the problem (rather than fixed) for otherwise the computational problem becomes easy (by, e.g., running the color coding algorithm).}
Since the case of  \emph{selfish} advertisers is of particular relevance in this context, we also initiate the study of mechanism design for the optimization problems introduced and consider the \emph{incentive-compatibility} of our algorithms, i.e., whether they can be augmented with payment functions so to work also with selfish advertisers.

For the version in which slot- and ad-dependant externalities cannot be factorized  and externalities occur in a window of size $c$, we prove that the optimization problem is in $P$, if $c$ is a constant. We consider the LP relaxation of the ILP describing the problem and prove that the integrality gap is 1. 
\begin{table*}[t]
\centering
\begin{tabular}{l|c|c|c|c|c|}
\cline{2-6}
                         & \multicolumn{2}{c|}{\textbf{\caa$(c)$}}                              & \multicolumn{2}{c|}{\textbf{\caa$(K)$}} &                               \multirow{2}{*}{\textbf{\csa$(c)$}} \\ \cline{2-5} 
                         & \textbf{nr}                          & \textbf{r}                         & \textbf{nr}                        & \textbf{r}                             &  \\ \hline
\multicolumn{1}{|l|}{\textbf{LB}} & APX-hard                    & \multirow{2}{*}{\rule{0ex}{13pt} APX-complete} & \multirow{2}{*}{\rule{0ex}{13pt} poly-APX-complete} & \multirow{2}{*}{\rule{0ex}{13pt} APX-complete} & \multirow{2}{*}{\rule{0ex}{13pt} P$\;^\star$} \\ \cline{1-2}
\multicolumn{1}{|l|}{\textbf{UB}} & \rule{0ex}{13pt}$\frac{\log(N)}{2\min\{N,K\}}^\star$ &                           &                           &        &                       \\ \hline
\multicolumn{1}{|l|}{\textbf{SP}} & \rule{0ex}{13pt} $\frac{\log(N)}{2\min\{N,K\}}^\star$     & $1/2$                       & $1/K$                     & $1/2$  & $1\;^\star$                         \\ \hline
\end{tabular}
\caption{Summary of our results: LB (UB, resp.) stands for lower (upper, resp.) bound on the approximation of the problem; the row SP, instead, contains the approximation guarantees we obtain with truthful mechanisms. Results marked by `$\star$' require $c=O(1)$. APX-completeness of a subclass of \caa$(c)$-nr is also given. (See the model for details on the notation.)}\label{tab:results}
\end{table*}

For the variant of the problem with factorized 
externalities, contextual ad-dependent externalities and window of $c$ slots, a distinction on the effects that empty slots have on users' behavior is useful. In a sort of whole page optimization fashion~\cite{wholepage}, we think of those slots as occupied by a \emph{special} (fictitious) ad used to refresh (e.g., by means of pictures) the user's attention. 

If the special ad cannot be used (or, equivalently, the user's attention cannot be reset) we prove that the allocation problem is poly-APX-complete whenever users have a ``large'' memory (i.e., the window equals the number of slots $K$). Specifically, we give an approximation preserving reduction from the Longest Path problem and design an 
approximation algorithm using several different ideas and sources of approximation; interestingly, its approximation guarantee matches the best known approximation guarantee for Longest Path. However, we prove that this algorithm cannot be used in any truthful mechanism and note that a simple single-item second price auction gives a weaker, yet close, truthful approximation. We complement the results for this model with the identification of tractable instances for which we provide an exact polynomial-time algorithm. For $c<K$ instead, we are unable to determine the exact hardness of approximating the problem in general. To the APX-hardness proof, we pair a number of approximation algorithms that assume constant $c$. The first, based on color coding \cite{colorcoding}, returns a non-constant approximation on any instance of SSA. The second assumes that the contextual graph is complete and returns a solution which (roughly) guarantees a $\gamma_{\min}^c$ fraction of the optimum social welfare, $\gamma_{\min}$ being the minimum edge weight in the graph. Interestingly, this algorithm shows the APX-completeness of the subclass of instances having   constant $\gamma_{\min}$ (we indeed further provide a hardness result for instances with complete contextual graphs). We believe the tight result for this subclass of instances to be quite relevant. In fact, complete contextual graphs are quite likely to happen in real-life: the results returned by a keyword search are highly related to one another, and, as such, each pair of ads has a non-null externality, however small.

If the special ad can be used, the problem becomes easier and turns out to be APX-complete, for any $c$. We first prove the problem with $c=K$ to be APX-hard, via a reduction from (a subclass of) ATSP (i.e., asymmetric version of TSP) and then surprisingly connect instances with $c<K$ to instances with $c=K$ by reducing the case with $c=1$ to the case with $c=K$ and \emph{binary} externalities (intuitively, the weights of the edges of the contextual graph can be either 0 or 1). We finally observe how a simple greedy algorithm cleverly uses the special ad to return $1/2$-approximate solutions and leads to a truthful mechanism. 

\section{Model} 
In a SSA we have $N$ ads and $K$ slots. We assume that each ad corresponds to an advertiser; this is w.l.o.g. from the optimization point of view. 
We denote each ad by $a_i$ with $i \in \N$, where $\N=\{1,\dots, N\}$ is the set of indices of the ads.
We introduce a fictitious ad, denoted by $a_\bot$, s.t., when allocated, the slot is left empty.
The $K$ slots are denoted by $s_m$ with $m \in \K$, $\K=\{1,\dots, K\}$ being the set of slot indices s.t. $s_1$ is the slot at the top of the page and $s_K$ is at the bottom.
We also have a fictitious slot, denoted by $s_\bot$ s.t. an ad allocated to $s_\bot$ is not displayed in the webpage. Each ad $a_i$ is characterized by: (\emph{i}) the \emph{quality} $q_i \in [0,1]$, i.e., the probability a user clicks on ad~$a_i$ when he observes it, irrespectively of other externalities; (\emph{ii}) the valuation $v_i \in \mathbb{R}^+$ advertiser~$i$ associates to his ad being clicked by a user.
The fictitious ad $a_\bot$ has $q_\bot = v_\bot = 0$.

A feasible allocation of ads to slots, denoted as $\theta$, consists of an ordered sequence of  ads $\theta=\langle a_1, \ldots, a_K\rangle$ s.t. the ads are ordered by increasing slot number, i.e., $a_1$ is allocated to the top slot, $a_K$ to the bottom one. Every ad $a_i$ can be allocated to at most one slot, whereas $a_\bot$ can be allocated to more than one slot. The set of all possible feasible allocations is denoted as $\Theta$. With a slight abuse of notation, we let (\emph{i}) $\theta(a_i)$ denote the index of the slot ad $a_i$ is allocated to, and (\emph{ii}) $\theta(s_m)$ denote the index of the ad allocated to $s_m$. 
Given $\theta \in \Theta$, the \emph{click through rate} of ad $a_i$, denoted as $CTR_i(\theta)$, is the probability ad $a_i$ is clicked by the user taking externalities into consideration. The optimal allocation $\theta^*$ is the one maximizing the \emph{social welfare}, namely: $\theta^* \in \arg \max_{\theta \in \Theta} SW(\theta)$, where $$SW(\theta)= \sum_{i \in \N} CTR_i(\theta) v_i.$$ A $1/\alpha$-approximate solution $\theta$ satisfies $SW(\theta) \geq SW(\theta^*)/\alpha$.

Typically, $CTR_i(\theta)$ defines how the quality $q_i$ of ad $a_i$ is ``perturbed'' by the externalities in terms of click probability. Accordingly, 
in general 
$CTR_i(\theta)=q_i \Gamma_i(\theta)$, $\Gamma_i(\theta)$ being a function encoding the effect of externalities. E.g., in the cascade model, $$\Gamma_i(\theta)=\Lambda_{\theta(a_i)} \prod_{l=1}^{\theta(a_i)-1}\gamma_{\theta(s_l)},$$ where $\Lambda_{\theta(a_i)}=\prod_{l=1}^{\theta(a_i)}\lambda_l$, $\lambda_m \in [0,1]$, called the \emph{factorized prominence} of~$s_m$, denotes the slot-dependant externality and $\gamma_i\ \forall i \in \N$, called \emph{continuation probability}, denotes the ad-dependent externality. (W.l.o.g., we assume $\Lambda_1=\lambda_1=1$.) 
Our conceptual contribution rests upon novel and richer ways to define $\Gamma_i(\theta)$, along three main dimensions.

The \emph{first dimension} concerns the \emph{user memory}, a.k.a. \emph{window}. We let $c$ be the number of ads displayed above $a_i$ in $\theta$, from $s_{\theta(a_i)-1}$ to $s_{\theta(a_i)-c}$, that affect $\Gamma_i(\theta)$. 
The \emph{second dimension} concerns a generalization of the externalities.
Here we propose two alternative families of externalities, called \emph{sa} (for slot-ad) and \emph{aa} (for ad-ad). The sa-externalities remove the factorization in slot- and ad--dependent externalities: i.e., 
$\lambda_m$ and $\gamma_i$ are substituted by parameters $\gamma_{m,j}\in [0,1]$, $m \in \K$ and $j \in \N$.
When the window is $c$, the CTR is defined as $CTR_i(\theta) = q_i \Gamma_i(\theta)$, where $$\Gamma_i(\theta)= \prod_{m = \max\{1,\theta(a_i)-c\}}^{\theta(a_i)-1} \gamma_{m,\theta(s_m)}.$$ 
This definition captures the situation in which an ad can affect the ads displayed below it in a different way according to the position in which it is displayed. %
For the aa-externalities, on the other hand, 
we preserve the factorization in $\lambda_m$ and $\gamma_i$, but redefine these latter parameters as $\gamma_{i,j}\in[0,1]$ where $a_j$ is the ad that is displayed in the slot just below $\theta(a_i)$.
It is convenient to see the $\gamma_{i,j}$'s as the weights of the \emph{contextual graph} $G=(\N,\E)$ where the direct edges $(i,j)$ weigh $\gamma_{i,j}>0$ and represent the way ad $a_i$ influences $a_j$. Note that non-edges of $G$ correspond to the pairs of ads $a_i$, $a_j$ s.t. $\gamma_{i,j}=0$. Here, with window $c$, $$\Gamma_i(\theta) = \Lambda_{\theta(a_i)} \prod_{l=\max\{1,\theta(a_i)-c\}}^{\theta(a_i)-1} \gamma_{\theta(s_l),\theta(s_{l+1})}$$ where $\Lambda_m$ is defined as above.
This definition captures the situation in which each ad can affect each other ad in a different way.

The \emph{third dimension} concerns the definition of $\gamma_{m,\bot}$ for the sa-externalities and $\gamma_{i,\bot}$ and $\gamma_{\bot, i}$ for the aa-externalities.
In the model \emph{with reset} we have $\gamma_{m,\bot} = 1$ for sa and $\gamma_{i,\bot} = \gamma_{\bot, i} = 1$ $\forall i \in \N \cup \{\bot\}$ for aa. This variant captures the situation in which slots can be distributed in the page in different positions (a.k.a., slates) and, in order to raise the user's attention, we can allocate a content, e.g. pictures, 
that nullifies the externality between the ad allocated before and after the content. In the model \emph{without reset}, $\gamma_{m,\bot} = 0$ for sa and $\gamma_{i,\bot}=\gamma_{\bot, i}=0$ $\forall i \in \N \cup \{\bot\}$ for aa, thus capturing the situation in which leaving a slot empty between two allocated slots does not provide any advantage. 

We let FNE$_x(c)$-y be the problem of optimizing the social welfare in our model with \emph{F}orward \emph{N}egative \emph{E}xternalities with window $c$, $x \in \{sa, aa\}$-ex\-ter\-nal\-i\-ties and y $ \in \{$r, nr$\}$ reset (r stands for reset; nr for no reset). When the value of y is not relevant for our results, we talk about FNE$_x(c)$.
We are interested in two particular subclasses of FNE$_{aa}(c)$, namely: (\emph{i}) subclass FNE$_{aa}^+(c)$-y, defined upon a complete contextual graph and such that $0 < \gamma_{\min} = \min_{i,j \in \N, i \neq j} \gamma_{i,j}$ and (\emph{ii}) subclass \B--FNE$_{aa}(c)$-y, where $\gamma_{i,j}$ can take values in $\{0,1\}$.

\subsection{Mechanism design} 
We use the theory of mechanism design to study the incentive-compatibility of our algorithms \cite{book}. 
A \emph{mechanism} ${M}$ is a pair $(A,P)$, where $A: (\mathbb{R}^+)^N \rightarrow \Theta$ is an algorithm that associates to any vector $\bv=(v_1, \ldots, v_N)$ of valuations a feasible outcome in $\Theta$
 (only valuations are private knowledge). The payment function $P_i:(\mathbb{R}^+)^{N} \rightarrow \mathbb{R}^+$ maps valuation vectors to monetary charges for advertiser~$i$. The aim of each advertiser is to maximize his own utility  $u_i(\bv, v_i) = CTR_i(A(\bv)) v_i - P_i(\bv)$. An advertiser could misreport his true valuation and declare $\hv_i \not = v_i$ when $u_i((\hv_i,\bv_{-i}), v_i) > u_i(\bv, v_i)$, $\bv_{-i}$ being the vector of the valuations of all the agents but $i$. We are then interested in truthful mechanisms. 
A mechanism is \emph{truthful} if for any $i \in \N$, $\bv_{-i} \in (\mathbb{R}^+)^{N-1}$, $v_i, \hv_i \in \mathbb{R}^+$, $u_i((\hv_i,\bv_{-i}), v_i) \leq u_i(\bv, v_i)$.

In this setting, a monotone algorithm \emph{must} be used in truthful mechanisms \cite{tardos}. 
Algorithm $A$ is monotone if for any $i \in \N$, $\bv_{-i} \in (\mathbb{R}^+)^{N-1}$, $CTR_i(A(\hv_i,\bv_{-i}))$ is non-decreasing in $\hv_i$. Important for our work is also the family of VCG-like mechanisms, a.k.a., \emph{Maximal In Range (MIR)} mechanisms. An algorithm $A$ is MIR if there exists $\Theta' \subseteq \Theta$ s.t. $A(\bv) \in \arg$ $\max_{\theta\in \Theta'} SW(\theta)$ $\forall \bv \in \mathbb{R}^{N}$~\cite{NisamRonen}. These algorithms can be augmented with a VCG-like payment so to obtain truthful mechanisms. (VCGs are MIR mechanisms wherein $\Theta' = \Theta$.) We are interested in mechanisms for which both $A$ and $P$ are computable in polynomial time. MIR mechanisms run in polynomial-time if the MIR algorithm does. 
As usual in the context of SSA, we adopt a pay-per-click payment scheme, i.e., we charge ${P_i(\bv)}/{CTR_i(A(\bv))}$ when a user clicks on~$a_i$. 

\section{FNE$_{sa}(c)$ is in $P$ for constant $c$}
Our presentation focuses on \csa$(1)$-nr to simplify the notation.
The more general cases when $c>1$ and the reset model is considered are easily obtainable by generalization from \csa$(1)$, but require a more cumbersome notation without significant new ideas (see discussion at the end of this section).
We first give the ILP formulation of \csa$(1)$-nr and prove that if there is an optimal fractional solution, then there are at least two feasible integral solutions with the same value of social welfare. Since it is well known, by LP theory, that the ellipsoid algorithm can be forced (in polynomial-time) to output an integral optimal solution, we are able to prove the following:
\begin{theorem}\label{thm:sainP}
For $c=O(1)$, there is a polynomial-time optimal algorithm for \csa$(c)$.
\end{theorem}
\noindent \csa$(1)$-nr can be formulated as following ILP: 
{\allowdisplaybreaks
\begin{align}
\max\sum_{m=2}^K  \sum_{i \in \N} \sum_{{j \in \N, j \not = i}}  \gamma_{m-1, j} q_i v_i x_{j,m,i} & + \sum_{i \in \N} x_{1,i} q_i v_i	 \nonumber \\
\textrm{subject to:} \hskip 11.4em & 	\nonumber \\
\sum_{m=2}^K \sum_{{j \in \N, j \not = i}} x_{j,m,i} + x_{1,i} \leq 1 			\quad & \quad \forall i \in \N 							\nonumber \\
x_{1,i} = \sum_{{j \in \N, j \not = i}} x_{i, 2, j} 						\quad & \quad \forall i \in \N				\nonumber 			\\
\sum_{j \in \N, j \not = i} x_{j,m,i} = \sum_{{j \in \N, j \not = i}} x_{i,m+1,j} 	\quad & \quad \forall i \in \N, 			\nonumber 			\\
														\quad & \quad 2 \leq m < K 			\nonumber 			\\
\sum_{i \in \N} x_{1, i} = 1 										\quad & \quad 						\label{eq:sumone}		\\
\sum_{j \in \N} \sum_{{i \in \N, i \not = j}} x_{j, m, i} = 1 				\quad & \quad \forall m \in \K\setminus\{1\}	\nonumber 			\\
x_{1,i} \in \{0,1\} 											\quad & \quad \forall i \in \N				\nonumber			\\
x_{j,m,i} \in  \{0,1\} 											\quad & \quad  \forall 2 \leq m \leq K, 		\nonumber			\\	
														\quad & \quad i, j \in \N, i \neq j 			\nonumber 
\end{align}
}
where $x_{j,m,i}=1$ iff $a_i$ is allocated to slot $s_m$, $m>1$, and $a_j$ is allocated to slot $s_{m-1}$; $x_{1,i}=1$ iff $a_i$ is allocated to $s_1$. The objective function and the constraints are rather straightforward and, hence, their description is omitted here. 

The next proposition proves Theorem \ref{thm:sainP} since it shows that we can solve the above ILP in polynomial-time, despite its similarities with the 3D-assignment, a well-known ${NP}$-hard problem.
\begin{proposition}\label{prop:csa1:poly}
The continuous relaxation of the above ILP always admits integral optimal solutions.
\end{proposition}

\begin{proof}
We show that, if there is an optimal fractional solution $x$, then there are at least two feasible integral solutions with the same value of social welfare. 
Specifically, we prove that $x$ is equivalent to a probability distribution over integral allocations $\theta = \langle a_1, \ldots, a_K\rangle $. The probability $\mathbb{P}(\theta)$ given to $\theta$ is:

\begin{align*}
\mathbb{P}(\theta) 			&	= \prod_{i=1}^K \mathbb{P}\left(	\theta(a_i)=s_i	\Big| \bigwedge_{j<i}	\theta(a_j) = s_j\right) 							\\ & = x_{1,1}\prod_{l=2}^K \frac{x_{l-1,l,l}}{\sum\limits_{m \geq l} x_{l-1,l,m}}.
\end{align*}
In order to show that $\mathbb{P}(\theta)$ is actually a probability distribution over allocations, we show that $\sum_{\theta \in \Theta} \mathbb{P}(\theta) = 1$. 

The proof is recursive. Let $\Theta'$ be the set of allocations $\theta$ with the same first $K-1$ ads.  
The allocations in $\Theta'$ differ only for the ad allocated to $s_K$. To fix the notation, for $\theta \in \Theta'$ let $\theta(s_l)=a_l$, for $l<K$. We have:
\begin{align*}
\sum_{\theta\in\Theta'} \mathbb{P}(\theta) & = x_{1,1}\prod_{l = 2}^{K-1} \left(\frac{x_{l-1,l,l}}{\sum_{m \geq l} x_{l-1,l, m}}\right) \sum_{h \geq K}\frac{x_{K-1,K, h}}{\sum\limits_{m \geq K} x_{K-1,K,m}}				&			\\
&= x_{1,1}\prod_{l = 2}^{K-1} \left(\frac{x_{l-1,l,l}}{\sum_{m \geq l} x_{l-1,l,m}}\right) \frac{\sum_{h \geq K}x_{K-1,K,h}}{\sum_{m \geq K} x_{K-1,K,m}}					\\	&	= x_{1,1}\prod_{l = 2}^{K-1} \left(\frac{x_{l-1,l,l}}{\sum_{m \geq l} x_{l-1,l,m}}\right). 	
\end{align*}

\noindent By applying recursively the same argument above from $\Theta'' \supset \Theta'$, the set of all allocations $\theta$ satisfying $\theta(s_l)=a_l$, for $l \leq K-2$, down to the set of allocations having only the same first ad, we have $\sum_{\theta:\theta(s_1)=a_1} \mathbb{P}(\theta) = x_{1,1}$. Since (\ref{eq:sumone}) forces $\sum_{i\in \N}x_{1,i}=1$, we have $\sum_{\theta \in \Theta} \mathbb{P}(\theta) = \sum_{i \in \N} x_{1,i} = 1$. This shows that $\mathbb{P}(\theta)$ is a well defined probability distribution. The proof concludes by observing that all integral solutions are indeed feasible.
\end{proof}

\noindent To solve the problem when $c>1$, we just need to modify the ILP and allow each variable $x$ to depend on $c+2$ indices to take into account the (at most) $c$ indices of all the ads that precede the ad of interest. The reset model for $c=1$ instead requires the introduction of $K$ additional variables for $a_\bot$ to be visualized in each slot (together with some constraints to fix each variable for $a_\bot$ to a slot). 

Theorem \ref{thm:sainP} implies that mechanism design becomes an easy problem for \csa$(c)$ and $c=O(1)$, since the optimal algorithm can be used to obtain a truthful VCG mechanism. 

\section{\caa$(K)$-nr is Poly--APX--Complete}
%
\subsection{Easy Instances}
As a warm-up, we identify a significant class of instances of \caa$(K)$-nr for which we can design a polynomial-time optimal algorithm. These instances are characterized by the fact that the underlying contextual graph is a DAG, thus modeling nearly oligopolistic markets in which the ads can be organized hierarchically. The idea of Algorithm \ref{alg:polytime_on_DAG} is that since DAGs can be sorted topologically in polynomial time then we can \emph{rename} the ads as $a_1,\ldots,a_N$ so to guarantee that for any pair of ads $a_i, a_j$, if $i<j$ then $(a_j,a_i)\notin \E$. We can then prove that we can focus w.l.o.g. on 
\emph{ordered} allocations $\theta$, i.e., for any pair of allocated ads $a_i, a_j$, with $i<j$, $\theta(a_i) \leq \theta(a_j)$. 
Consider an unordered $\theta$ and let $a_i$ be the first
ad (from the top) for which there exists $a_j$, $i<j$, such that
$\theta(a_i) > \theta(a_j)$. Since $\gamma_{j,i}=0$ then all the ads $a_k$ s.t. $\theta(a_k) \geq \theta(a_i)$ have $CTR_k(\theta)=0$ and, therefore, we can prune $\theta$ of (i.e., substitute with $a_\bot$) $a_i$ and all the subsequent ads without any loss in the social welfare. But then in the class of ordered allocations, the optimum has an optimal substructure and we can use dynamic programming.
Let $D[i, m]$ be the value of the optimal ordered allocation that uses only slots $s_m, \ldots, s_K$ and allocates ad $a_i$ in $s_m$. It is not hard to see that $D[i,m] = \Lambda_m q_i v_i + \max_{j>i} \gamma_{i,j} D[j,m+1]$ and that the optimum is $\max_{i \in [N]} D[i,1]$. In the pseudo-code of the algorithm, we simply construct the table $D$ after the topological sort of the contextual graph (with renaming of the ads) is done. The algorithm runs in time  $O(KN^2)$.

\begin{algorithm}
\begin{algorithmic}[1]
\STATE $\textsc{TopologicalSort}(G)$ \label{s:par2o}
\FOR{all $m \leq K$} \label{s:lrow}
	\STATE $D[N,m] = \Lambda_m q_N v_N$ \label{s:endlrow}
\ENDFOR
\FOR{all $i \leq N$} \label{s:lcol}
	\STATE $D[i,K] = \Lambda_K q_i v_i$ \label{s:endlcol}
\ENDFOR
\FOR{$i = N -1$ to $1$} \label{s:table}
		\FOR{$m = K-1$ to $1$} 
		\STATE $D[i,m] = \Lambda_m q_i v_i + \max_{j>i}  \gamma_{i,j} D[j,m+1]$ \label{s:endtable}
	\ENDFOR	
\ENDFOR
\RETURN{$(\max_{i \in [N]} D[i,1])$} \label{s:max1slot}
\end{algorithmic}
\caption{}\label{alg:polytime_on_DAG}
\end{algorithm}

\noindent Since social welfare maximization is a utilitarian problem, and given that the algorithm above is optimal we can use the VCG mechanism to obtain a polynomial-time optimal truthful mechanism.

%

\subsection{Hardness} 
We now prove the hardness of approximating \caa$(K)$-nr.
\begin{theorem}
\caa$(K)$-nr is poly--APX--hard.
\end{theorem}
\begin{proof}
We reduce from the Longest Path problem. An instance of the Longest Path problem consists of a direct graph $G' = (T,A)$ where $T$ is the set of vertices of the graph and $A \neq \emptyset$ is the set of unweighted edges. The problem demands to compute a \emph{longest simple path}, i.e., a maximum length path that visits each vertex of the graph at most once.
This problem is poly--APX--complete~\cite{LONGESTPATH} and the best known asymptotic approximation is ${\log |T|}/{|T|}$. 
From an instance $G'=(T,A)$ of Longest Path we obtain an instance of \caa$(K)$-nr as follows. 
For each vertex $t_i\in T$ we add an ad $a_i$, with $q_i=v_i=1$ and for each directed arc $(t_i,t_j)\in A$ we add an arc $(i,j)$ in $\E$.
Furthermore, we set $\gamma_{i,j} = 1$ if $(i,j)\in \E$ and $\gamma_{i,j}=0$ otherwise.
Finally, we set $N=K=|T|$ and $\Lambda_m=1$, $\forall m \in [K]$.

Given an ordered sequence of vertices  $\rho = ( t_1, t_2, \ldots, t_N )$, we denote as $len(\rho)$ the length of the path that starts in $t_1$ and visits the nodes in $\rho$ till the first node $t_j$ s.t. $(t_j, t_{j+1}) \not \in A$ is reached.
Let us denote as $\rho^*$ the sequence that describes the longest path in $G'$ and as $\theta^*$ the allocation that maximizes the social welfare in the instance of \caa$(K)$-nr defined upon $G'$.
It is easy to check that $len(\rho^*) = SW (\theta^*) - 1$.
Indeed, $\theta^*$ allocates sequentially from the first slot the ads that correspond to the vertices composing the longest path. Conversely, we can transform an allocation $\theta$ into a sequence of vertices $\rho$ just by substituting the ads with their corresponding vertices until the first $a_\bot$ in $\theta$ is found. Thus, we have that for $\theta$ and the corresponding $\rho$ it holds $len(\rho) = SW(\theta) - 1$.

Consider a generic $\alpha$-approximate allocation $\theta_{\alpha}$ for \caa$(K)$-nr: $SW(\theta_{\alpha}) \geq \alpha SW(\theta^*)$. 
Since $A$ is non-empty, there is a solution $\theta_2$ to \caa$(K)$-nr of social welfare at least $2$. 
Let $\theta_{\beta}$ denote the solution in $\{\theta_\alpha, \theta_2\}$ with maximum social welfare. As $\theta_\alpha$ is an $\alpha$-approximate solution so is $\theta_{\beta}$. By letting $\rho_{\beta}$ denote the path constructed from $\theta_{\beta}$ as described above, we prove that the reduction preserves the approximation (up to a constant factor): 
$len(\rho_{\beta}) = SW(\theta_{\beta}) - 1 \geq \frac{1}{2} SW(\theta_{\beta}) \geq \frac{\alpha}{2} SW(\theta^*) =  \frac{\alpha}{2} \left(len\left(\rho^*\right) + 1\right) \geq \frac{\alpha}{2} len(\rho^*).$ \qed
%
\end{proof}
%

\subsection{Approximation algorithm} 
We show that the problem is in poly--APX, with an approximation ratio that is asymptotically the same as the best guarantee known for Longest Path. Our  algorithm combines the Color Coding (CC) algorithm \cite{colorcoding} together with three approximation steps.

Let $C$ be a set containing $K$ different colors. CC is a random algorithm, randomly assigning colors from $C$ to the ads, and then finding the best \emph{colorful} (i.e., no pair of ads has the same color) allocation. 
To find the best colorful allocation, given a random coloring we do the following. For $S\subseteq C$, we define $(S,a_i)$ as the set of partial allocations with the properties of having the same number $|S|$ of allocated ads (each colored with a different color of $S$) in the first $|S|$ slots and having ad~$a_i$ in slot~$s_{|S|}$. We start from $S=\emptyset$ where no ad is allocated. Then, allocating one of the ads in the first position, we add one color to $S$ until $S=C$. Iteratively, the algorithm extends the allocations in $(S, a_i)$ appending a new ad, say $a_j$, with a color not in $S$ in slot $s_{|S|+1}$ obtaining $(S \cup \{o_j\}, a_j)$ where $o_j$ is the color of $a_j$. Each partial allocation in $(S, a_i)$ is characterized by the values of $SW$ and $\Gamma_i$. We can safely discard all the Pareto dominated partial allocations: given two allocations $\theta_1$ and $\theta_2$ in $(S, a_i)$, we say that $\theta_2$ is Pareto dominated by $\theta_1$ iff $SW(\theta_1) \geq SW(\theta_2)$ and $\Gamma_i(\theta_1) \geq \Gamma_i(\theta_2)$. However, there is no guarantee that the number of allocations in $(S,a_i)$ is polynomially bounded and, in principle, all the generated $O(N^K)$ partial allocations may be Pareto efficient. The complexity per coloring is $O(2^KN^{K+1}K^2)$. CC generates $e^K$ random colorings, but it can be derandomized with a cost of $\log^2(N)$ and a total complexity $O((2e)^K K^2 N^{K+1} (\log N)^2)$. To make the algorithm polynomial, we apply three approximation steps. Initially, we briefly sketch these three approximations and, subsequently, we provide the details. Firstly, we run CC  over a reduced number $K'$ of slots where $K' = \min(\lceil\log (N)\rceil, K)$. Secondly, we discard all the allocations $\theta$ in which the probability to click on the last allocated ad is smaller than a given $\delta \in [0,1]$. Finally, we discretize the $\gamma_{i,j}$'s. We prove in the following that the running time is indeed polynomial and the approximation ratio is $(1-\delta)(1-\epsilon)\frac{\log (N)}{2 \min \{N,K\}}$, $\epsilon$ controlling the granularity of the $\gamma_{i,j}$ discretization. All the three approximations are necessary in order to obtain a polynomial-time algorithm. This algorithm is not monotone as we show below. However, a simple $1/K$-approximate truthful mechanism can be obtained, via a single-item second price auction. From here on, we provide the details of the algorithms and we prove its approximation ratio.

\smallskip \noindent \\underline{\emph{Approximation 1}.} We apply CC over a reduced number $K'$ of slots, where $K' = \min(\lceil\log (N)\rceil, K)$, implying the following approximation ratio.

\begin{proposition} \label{p:K'}
Given $\theta^*$, the optimal allocation over $K$ slots, and $\theta^*_{K'}$, the optimal allocation over the first $K'\leq \min\{N,K\}$ slots, we have $SW\left(\theta^*_{K'}\right) \geq \frac{1}{2} \frac{K'}{\min\{N,K\}}  SW\left(\theta^*\right)$.
\end{proposition}

\begin{proof}
We partition $K'' = \min\{N,K\}$ slots in groups of $K'$ consecutive slots. There could be remaining slots that will constitute the last group with less then $K'$ slots. The number of groups in which the $K$ slots are divided is $NG=\lceil \frac{K''}{K'} \rceil$. Let $G_i = \{(i-1)K'+1, \ldots, \min(i K', K)\}$, for $i \in [NG]$, be the $i$-th group of indices of $K'$ slots.

We let $SW(\theta|G_i) = \sum_{m \in G_i} \Lambda_m \Gamma_{\theta(m)}(\theta) q_{\theta(m)} v_{\theta(m)}$, for any $\theta \in \Theta$. Since $SW(\theta^*) = \sum_{i=1}^{NG} SW(\theta^*|G_i)$, there must exist a group $G_i$ s.t. $SW(\theta^*|G_i) \geq \frac{1}{NG} SW(\theta^*)$. Observing that $\lceil \frac{K''}{K'} \rceil \leq \frac{K''}{K'} + 1$ and $K' \leq K''$ we get $SW(\theta^*|G_i) \geq \frac{K'}{2K''} SW(\theta^*)$. The proof concludes by noting that, by optimality, $SW(\theta^*_{K'}) \geq SW(\theta^*|G_i)$. 
\end{proof}

\smallskip \noindent \underline{\emph{Approximation 2}.} In CC, we discard allocations $\theta$ in which $\Gamma_i(\theta)$ of the last allocated ad $a_i$, $i \in [N]$, is less than a given $\delta \in [0,1]$, implying the following approximation ratio.

\begin{proposition} \label{p:>=delta}
Given $\theta^*_{K'}$, the optimal allocation over $K'$ slots, and $\theta^{\delta}_{K'}$ the optimal allocation among the allocations $\theta \in \Theta$ where the last allocated ad $a_i$, $i \leq N$, satisfies $\Gamma_i(\theta) \geq \delta$, we have $SW\left(\theta^{\delta}_{K'}\right) \geq \left(1-\delta\right) SW\left(\theta^*_{K'}\right)$.
\end{proposition}
\begin{proof}
Consider the allocation $\theta^*_{K'}$ and assume that the last ad satisfying $\Gamma_i(\theta^*_{K'})\geq \delta$ is the one in slot $s_l$. 
Recalling the notation $SW(\theta|S)$ for $S \subseteq [K]$, provided in the proof of Proposition~\ref{p:K'}, by optimality of $\theta^*_{K'}$ we have $SW(\theta^*_{K'}) \geq \frac{1}{\Gamma_{\theta^*_{K'}(l+1)}} SW(\theta^*_{K'}|\{l+1,\ldots,K\})$. Indeed, on the r.h.s. we have a lower bound on the social welfare that the ads allocated by $\theta^*_{K'}$ in slots $s_{l+1}, \ldots, s_{K'}$ would have if shifted to the first slot. If this were bigger than $SW(\theta^*_{K'})$ then $\theta^*_{K'}$ would not be optimal. But then since $\Gamma_{\theta^*_{K'}(l+1)} < \delta$, we have
$\delta SW(\theta^*_{K'}) \geq SW(\theta^*_{K'}|\{l+1,\ldots,K\})$.

Finally we have that $\theta^{\delta}_{K'}$, the allocation that removes from $\theta^*_{K'}$ the ads allocated from $s_{l+1}$ to $s_{K'}$, has $SW(\theta^{\delta}_{K'}) = SW(\theta^*_{K'}) - SW(\theta^*_{K'}|\{l+1,\ldots,K\})
\geq SW(\theta^*_{K'}) - \delta SW(\theta^*_{K'}) = (1-\delta) SW(\theta^*_{K'})$. 
\end{proof}

\smallskip \noindent \underline{\emph{Approximation 3}.} In CC, we use rounded values for $\gamma_{i,j}$.
More precisely, we use $\lfloor \frac{1}{\tau}\log \frac{1}{\gamma_{i,j}} \rfloor$ in place of $\log \frac{1}{\gamma_{i,j}}$, where the normalization constant $\tau$ is defined below. 
The constraint due to Proposition~\ref{p:>=delta} is now a capacity constraint of the form $\sum_{m \in [K]: m < l} \lfloor \frac{1}{\tau}\log \frac{1}{\gamma_{\theta(m),\theta(m+1)}} \rfloor \leq \lfloor \frac{1}{\tau}\log \frac{1}{\delta} \rfloor$.
Notice that, with rounded values, the capacity can assume a finite number of values (i.e., $\lfloor \frac{1}{\tau}\log \frac{1}{\delta} \rfloor$) and therefore we can now bound the number of allocations to be stored in $(S,a_i)$.
More precisely, for each value of capacity, we can discard all the allocations except one maximizing the social welfare measured with rounded values. This step has the following consequences on the approximation guarantee.

\begin{proposition} 
Given $\theta^\delta_{K'}$, defined as in Proposition \ref{p:>=delta}, and $\theta^{\delta\epsilon}_{K'}$, the optimal allocation when the rounding procedure is applied, we have that, choosing $\tau = \frac{1}{K'}\log \frac{1}{1 - \epsilon}$, $SW\left(\theta^{\delta\epsilon}_{K'}\right) \geq \left(1-\epsilon\right) SW\left(\theta^{\delta}_{K'}\right)$.
\end{proposition}
\begin{proof} 
Let $\xi^{x}_{m, m+1}$ be a shorthand for $\log \frac{1}{\gamma_{\theta_{K'}^{x}(m),\theta_{K'}^{x}(m+1)}}$ and $x(i)$ be a shorthand for $\theta^x_{K'}(a_i)$, for $x \in \{\delta \epsilon,\delta \}$.
By definition:
\begin{align*}
SW\left(\theta^{\delta\epsilon}_{K'}\right) & = \sum\limits_{i \in [N]} \Lambda_{{\delta\epsilon}(i)} \Gamma_{i}\left(\theta^{\delta\epsilon}_{K'}\right) q_i v_i 
\\ & = \sum\limits_{i \in [N]} \Lambda_{{\delta\epsilon}(i)} 
\prod_{m < {\delta\epsilon}(i)} 2^{-\xi^{\delta \epsilon}_{m, m+1}} 
q_i v_i.
\end{align*}

\noindent Since $\xi^{\delta \epsilon}_{m, m+1} \leq \tau (\lfloor \frac{1}{\tau}\xi^{\delta \epsilon}_{m, m+1} \rfloor +1)$, we then have 
\begin{align*}
SW\left(\theta^{\delta\epsilon}_{K'}\right) & \geq \sum\limits_{i \in [N]} \Lambda_{{\delta \epsilon}(i)} \prod_{m < {\delta\epsilon}(i)} 2^{-\tau\left( \left \lfloor \frac{1}{\tau} \xi^{\delta \epsilon}_{m, m+1} \right \rfloor + 1  \right)} q_i v_i\\
& \geq \sum\limits_{i \in [N]} \Lambda_{{\delta}(i)} \prod_{m < {\delta}(i)} 2^{-\tau\left( \left \lfloor \frac{1}{\tau}\xi^{\delta}_{m, m+1} \right \rfloor + 1  \right)} q_i v_i,
\end{align*}

\noindent where the latter inequality follows from optimality of $\theta^{\delta}_{K'}$. Given that $\lfloor y \rfloor \leq y$ we can conclude that $SW\left(\theta^{\delta\epsilon}_{K'}\right)$ is bounded from below by:
\begin{align*} 
& \sum\limits_{i \in [N]} \Lambda_{\delta(i)}  \left(\prod_{m < {\delta}(i)} 2^{\log \gamma_{\theta_{K'}^{\delta}(m),\theta_{K'}^{\delta}(m+1)} -\tau }\right) q_i v_i\\
& \geq 2^{-K'\tau}\cdot \sum\limits_{i} \Lambda_{{\delta}(i)}  \Gamma_i\left(\theta_{K'}^\delta\right) q_i v_i\\
& = (1-\epsilon)\cdot \sum\limits_{i} \Lambda_{{\delta}(i)} \Gamma_i\left(\theta_{K'}^\delta\right) q_i v_i =\left(1-\epsilon\right) SW\left(\theta^{\delta}_{K'}\right). 
\end{align*}
\noindent This concludes the proof.  
\end{proof}

The approximation ratio of the algorithm is thus $(1-\delta)(1-\epsilon)\frac{\log (N)}{2 \min \{N,K\}}$, asymptotically the same as the best known approximation ratio of the Longest Path once $N=K$. The complexity instead can be derived as follows. The maximum number of allocations that can be stored in each $(S, a_i)$ is  $O(\frac{\log \frac{1}{\delta}}{\tau})$ with $\tau = \frac{\log \frac{1}{1 - \epsilon}}{K'}$ thanks to dominations. Thus, given that $\log(\frac{1}{1-\epsilon}) \rightarrow \epsilon$ as $\epsilon \rightarrow 0$, the number of elements is $O(K' \frac{1}{\epsilon})$. Thus, the complexity when $K' = \log (N)$ is $O((2e)^{\log (N)} \frac{1}{\epsilon}  \log(\frac{1}{\delta})N^2  \log^4 (N))=O(\frac{1}{\epsilon\delta}N^3  \log^4 (N))$.

Notice that all the three above approximations are necessary in order to obtain a polynomial--time algorithm. Approximation~2 and Approximation~3 allow us to bound the number of the allocations stored per pair $(S,a_i)$ and would lead, if applied without Approximation~1, to a complexity $O((2e)^KK^2N^2\log^2(N)\frac{1}{\epsilon\delta})$. Notice also that, without Approximation~2, the possible values for the capacity are not upper bounded. Approximation~1 allows us to remove the exponential dependence on $K$ and to obtain polynomial complexity.

\subsubsection*{Non--monotonicity of the approximation algorithm} 


In this section we prove that the 
algorithm is not monotone and therefore we cannot augment it with a payment function to obtain a truthful mechanism.

Let us initially consider the case where Approximation~1 is not used, therefore all the $K$ slots can be allocated. We will discuss below how to extend the proof to the case where Approximation~1 is used.

Consider the following instance of \csa$(K)$-nr:
\begin{itemize}
\item $K=3$ slots;
\item $N=4$ ads, where $q_1 v_1 = 2^{2\tau} \frac{\Lambda_2 -\Lambda_3 2^{-6\tau}}{\Lambda_2 - \Lambda_3} + 3$, $q_2 v_2 = x$, $q_3 v_3 = q_4 v_4 = 1$, where $\tau$ is the generic rounding factor of Approximation~3;
\item the contextual graph is s.t. $\gamma_{i,j}=0$ $\forall i,j \in [N]$ except: $\gamma_{1,2} = 2^{\left(-4 + \phi\right) \tau}$, $\gamma_{1,3} = 2^{- \tau}$, $\gamma_{2,4} = 2^{-\tau}$, $\gamma_{3,2} = 2^{- \tau}$. $\phi$ is a small number;
\item the rounded capacity $\left \lfloor \frac{\log \frac{1}{\gamma_{i,j}}}{\tau} \right \rfloor = +\infty$ $\forall i,j \in [N]$ except: $\left \lfloor \frac{\log \frac{1}{\gamma_{1,2}}}{\tau} \right \rfloor = 3$, $\left \lfloor \frac{\log \frac{1}{\gamma_{1,3}}}{\tau} \right \rfloor = 1$, $\left \lfloor \frac{\log \frac{1}{\gamma_{2,4}}}{\tau} \right \rfloor = 1$, $\left \lfloor \frac{\log \frac{1}{\gamma_{32}}}{\tau} \right \rfloor = 1$.
\item the $K$ colours are $\{o_1, o_2, o_3\}$.
\end{itemize}

The product $q_1 v_1$ has been chosen s.t., when $x$ is in the neighbourhood of $2^{2\tau} \frac{\Lambda_2 - \Lambda_3 2^{-4\tau}}{\Lambda_2 - \Lambda_3}$, $a_1$ is always allocated in the first slot. Thus, we can focus only on the colouring that assigns colour $o_1$ to $a_1$, $o_2$ to $a_2$ and $o_3$ to $a_3$ and $a_4$. Indeed, with this colouring the two longest path of the contextual graph are colourful, i.e. the unique two colourful allocations are $\theta_1= ( a_1, a_3, a_2 )$ in the set $(\{o_1,o_2,o_3\}, a_2)$ and $\theta_2= ( a_1, a_2, a_4 )$ in the set $(\{o_1,o_2,o_3\}, a_4)$.

Notice that, with this colouring, all the allocations where there is a pair of ads $(a_i,a_j)$ with $\gamma_{i,j}=0$ are infeasible, not satisfying the capacity bound.
We will now prove that the approximation algorithm is not monotone with respect to $a_2$.

Let us denote by $\widetilde{SW}$ the social welfare computed on the basis of the rounded values.
It is easy to check that the following hold: $\widetilde{SW}(\theta_1) = 2^{2\tau} \frac{\Lambda_2 -\Lambda_3 2^{-6\tau}}{\Lambda_2 - \Lambda_3} + 3 + \Lambda_2 2^{-4\tau} x + \Lambda_3 2^{-6\tau}$ and $\widetilde{SW}(\theta_2) = 2^{2\tau} \frac{\Lambda_2 -\Lambda_3 2^{-6\tau}}{\Lambda_2 - \Lambda_3} + 3 + \Lambda_2 2^{-\tau} + \Lambda_3 2^{-4\tau} x$.
Notice that the rounded $CTR_2$ in $\theta_2$ is always greater than the one in $\theta_1$, given $\Lambda_2 \geq \Lambda_3$, while $CTR_2(\theta_1) = \Lambda_3 2^{-2 \tau} > \Lambda_2 2^{\left(-4 + \phi\right)\tau} = CTR_2(\theta_2)$ when $\frac{\Lambda_2}{\Lambda_3} < 2^{2\tau - \phi \tau}$.

We have that $\widetilde{SW}(\theta_1) > \widetilde{SW}(\theta_2)$ when $x > 2^{2\tau} \frac{\Lambda_2 - \Lambda_3 2^{-4\tau}}{\Lambda_2 - \Lambda_3}$.
Thus $a_2$ gets a lower CTR by increasing her bid, which proves that the algorithm is not monotone.

The example can be extended also to the case where Approximation~1 is applied introducing ads with $qv=0$ and $\gamma_{i,j} = 0$, s.t. $\log N = K$.

\section{\caa$(K)$-r is APX-complete}

In this section we will prove the APX-hardness of \caa$(K)$-r and provide a $1/2$-approximation algorithm. 

\subsection{Hardness}
In this section we prove that
\caa$(K)$-r is APX--hard. 

\begin{theorem}\label{thm:CNFE_inapproximability}
\caa$(K)$-r cannot be approximated within a factor of $\frac{1}{1+\alpha}$, for $\alpha < \frac{1}{412}$, unless $P=NP$.
\end{theorem}

\begin{proof}
We reduce from the Asymmetric TSP with weights in $\{1,2\}$, hereinafter denoted as $ATSP(1,2)$.
The $ATSP(1,2)$ problem demands finding a minimum cost Hamiltonian tour in a complete directed weighted graph $G'=(T,A)$ where $T$ is the set of nodes of $G'$, $A$ is the set of edges and the weight function $w_{i,j}\in \{1,2\}$ for all edges $(i,j)\in A$.
$ATSP(1,2)$ cannot be approximated in polynomial time within a factor of $\frac{1}{1+\beta}$, with $\beta<1/206$~\cite{ATSP1_2}. 
Below, we denote as $\tau$ a solution of an $ATSP(1,2)$ instance, as $cost(\tau)$ its cost and as $\tau^*$ the optimal tour. 

Given an instance of $ATSP(1,2)$ on graph $G'=(T,A)$ we construct an instance of \caa$(K)$-r as follows: (\emph{i}) for each vertex $t_i \in T$ we generate an ad $a_i$ with $q_i=v_i=1$, then we have $N=|T|$; (\emph{ii}) the contextual graph is $G=([N],\E)$, where $(i,j)\in \E$ iff $w_{i,j}=1$; (\emph{iii}) for all $(i,j)\in \E$, $\gamma_{i,j} = 1$; and finally (\emph{iv}) the number of slots is equal to the cost of the optimal tour $\tau^*$ in $ATSP(1,2)$, i.e. $K = cost(\tau^*)$.
We will show at the end of the proof how we can deal with the fact that we do not know $cost(\tau^*)$. Observe that with $K=cost(\tau^*)$, we have $SW(\theta^*) = N$, $\theta^*$ denoting the optimal solution of the \caa$(K)$-r instance constructed.
The definition of the reduction is completed by observing that an allocation $\theta$ for the \caa$(K)$-r that allocates all the $N$ ads can be easily mapped back to a tour $\tau$ for the $ATSP(1,2)$ by simply substituting the ad with the corresponding vertex of the graph $G'$.

Let us suppose for the sake of contradiction that there exists a $\frac{1}{1+\alpha}$-approximate algorithm for \caa$(K)$-r, with $\alpha<\frac{\beta}{2}<\frac{1}{412}$.
Let $\theta_{\alpha}$ be the $\frac{1}{1+\alpha}$--approximate solution returned by such an algorithm, i.e., $SW(\theta_{\alpha}) \geq \frac{1}{1+\alpha} SW(\theta^*) = \frac{N}{1+\alpha}$.
It is easy to check that $\theta_\alpha$ consists of $\lceil \frac{N}{1+\alpha} \rceil$ ads, each providing a contribution of 1 to the social welfare, while there are $SW(\theta^*) - \lceil \frac{N}{1+\alpha} \rceil$ ads that w.l.o.g. we can consider empty.
Moreover, being $\alpha < 1$, $\frac{N}{1+\alpha} \geq cost(\tau^*) - \frac{N}{1+\alpha}$ holds.
For the sake of conciseness, hereinafter we omit the ceiling notation.
Let $\tau_\beta$ be the tour obtained from $\theta_\alpha$.
We state that in $\tau_{\beta}$ there are, at least, $\frac{2N}{1+\alpha} - cost(\tau^*) - 1$ edges of weight 1.
Divide the ads allocated in $\theta_\alpha$ in two sets: the $\frac{N}{1+\alpha}$ allocated ads $a_i\ i \in [N]$ and $a_{\bot}$.
Allocate in alternation one of the $\frac{N}{1+\alpha}$ ads $a_i$, with $i \in [N]$, and one of the $cost(\tau^*) - \frac{N}{1+\alpha}$ ads $a_\bot$.
When the slot index $2(cost(\tau^*) - \frac{N}{1+\alpha})$ is reached, the available $a_\bot$ are finished, thus, in the following $cost(\tau^*) - 2(cost(\tau^*) - \frac{N}{1+\alpha}) = \frac{2N}{1+\alpha} - cost(\tau^*)$ slots, only non-fictitious ads $a_i$, $i \in [N]$, are consecutively allocated (no slots are left empty).
This means that in $\theta_\alpha$, where the ads are disposed in a different way, we still have the guarantee that there are $\frac{2N}{1+\alpha} - cost(\tau^*) - 1$ pairs of consecutive ads $(a_i,a_j)$ s.t. $\gamma_{i,j}=1$.
Thus, in the tour $\tau_{\beta}$ there are, at least, $\frac{2N}{1+\alpha} - cost(\tau^*) - 1$ edges of weight 1.
Therefore, given that a tour is composed of $N$ edges, in $\tau_{\beta}$ there can be at most $N -\frac{2N}{1+\alpha} + cost(\tau^*) + 1$ edges of weight 2.
The length of $\tau_{\beta}$ is upper-bounded by $cost(\tau_{\beta}) \leq \frac{2N}{1+\alpha} - cost(\tau^*) - 1 + 2 (N -\frac{2N}{1+\alpha} + cost(\tau^*) + 1) = cost(\tau^*) + \frac{2N\alpha}{1+\alpha} + 1$.
Now we can state:
$
cost(\tau_\beta) \leq cost(\tau^*) + \frac{2\alpha N}{1+\alpha} + 1
 \leq cost(\tau^*)+2\alpha N 
 \leq cost(\tau^*)+2\alpha\, cost(\tau^*)
 = (1+2\alpha)\, cost(\tau^*)
 < (1+\beta)\, cost(\tau^*),
$
where: (i) the second inequality holds for $N\geq \frac{1+\alpha}{2\alpha^2}$; (\emph{ii}) the third inequality holds since $N\leq cost(\tau^*)$ and (\emph{iii}) the last inequality holds since, by assumption, $\alpha<\frac{\beta}{2}$.
Thus, for the instances where $N\geq \frac{1+\alpha}{2\alpha^2}$ if there were an algorithm that $\frac{1}{1+\alpha}$--approximates \caa$(K)$-r with $\alpha<\frac{1}{412}$, there would be a  $\frac{1}{1+\beta}$ approximation of $ATSP(1,2)$ with $\beta < \frac{1}{206}$. We obtained an absurd.

We finally show that we can deal with the non existence of the oracle returning $cost(\tau^*)$. For all the instances of $ATSP(1,2)$ with $N$ vertices, $N \leq cost(\tau^*) \leq 2N$. So, 
we run the polynomial $\frac{1}{1+\alpha}$--approximation algorithm of \caa$(K)$-r for all the values $K=m$ with $m \in \{N \ldots, 2N\}$, obtain $m$ 
tours $\tau_{\beta}^m$ and 
set $\tau_{\beta} = \arg\min_{ m \in \{N,\ldots,2N\}} cost(\tau_{\beta}^m)$,  guaranteeing $cost(\tau_{\beta})\leq cost(\tau_{\beta}^{cost(\tau^*)})$.   
\end{proof} 

\subsection{$\frac{1}{2}$-Approximate Greedy Algorithm for \caa$(c)$-r, for any $c$}. The algorithm  orders the ads in nonincreasing order of $q_i v_i$ and allocates them in the odd slots, starting from the one with the highest product; even slots are left empty. 

\begin{proposition}\label{prop:greedy}
The greedy algorithm above is $\frac{1}{2}$-approxi-mate for \caa$(c)$-r, for any $c$.
\end{proposition}
\begin{proof}
Let $\theta_{.5}$ be the allocation obtained by the algorithm. We want to prove that $SW(\theta_{.5}) \geq SW (\theta^*)/2$. W.l.o.g., rename the ads so that $q_1 v_1 \geq q_2 v_2 \geq \ldots \geq q_N v_N$. Let $K'= \left \lceil {K/2} \right \rceil$. We have $SW(\theta_{.5}) = \sum_{m \in [K']} \Lambda_{2m-1} q_m v_m$. On the other hand, $SW(\theta^*) \leq \sum_{m \in [K]} \Lambda_m q_m v_m$. Since $\Lambda_i q_i v_i \geq \Lambda_{i+1} q_{i+1} v_{i+1}$, we have $\Lambda_i q_i v_i \geq {1/2} \sum_{m=i,i+1} \Lambda_m q_m v_m$. We conclude: 
\begin{align*}
SW(\theta_{.5}) = 	&	\sum_{m\in [K']} \Lambda_{2m-1} q_m v_m \geq 					\\
				&	\sum_{m\in [K']} \Lambda_{2m-1} q_{2m-1} v_{2m-1} \geq 			\\
				&	{1/2} \sum_{m\in [K]} \Lambda_m q_m v_m \geq SW(\theta^*)/2. \qedhere 
\end{align*}
\end{proof} 
\noindent The greedy algorithm above is a MIR, range $\Theta'$ being all the allocations that leave even slots empty. The solution output is indeed the one guaranteeing maximum social welfare in $\Theta'$. We therefore have proved the existence of a $1/2$-approximate truthful polynomial-time mechanism for \caa$(c)$-r.

\section{FNE$_{aa}$(c) is APX-hard} \label{sec:c<K.APX-hard}

We now prove that FNE$_{aa}(1)$-r (Proposition \ref{prop:FNE_1-r}) and FNE$_{aa}(1)$-nr (Proposition \ref{prop:FNE_1-nr}) are APX-hard.
First we state two auxiliary lemmata. 
Hereinafter, for the sake of notation, we will denote as $SW_1(\theta)$ and $SW_K(\theta)$ the objective function of \B--FNE$_{aa}(1)$-r and \B--FNE$_{aa}(K)$-r, respectively.

\begin{lemma}\label{lemma:no_gamma_0}
Let $\theta$ be an allocation (possibly containing empty slots) 
and let $\theta'$ be the allocation obtained from $\theta$ by 
replacing, for each pair $(a_{i-1},a_{i})$ in $\theta$ such that $\gamma_{i-1,i}=0$, ad $a_{i-1}$ with $a_\bot$. Then $SW_1(\theta)=SW_1(\theta')$.
\end{lemma}
\begin{proof}
Let $(a_{i-1},a_i)$ be the first pair of ads in $\theta$ with the property that $\gamma_{i-1,i}=0$, and let $\theta''$ be the allocation obtained from $\theta$ by substituting $a_{i-1}$ with $a_\bot$.
Let $SW_1^A (\theta)=\sum_{j=1}^{i-2} CTR_j(\theta)v_j$ and $SW_1^B (\theta)=\sum_{j=i+1}^{K} CTR_j(\theta)v_j$ denote the contributions to the $SW$ of the ads allocated, respectively, above and below the pair $(a_{i-1},a_i)$.
We can write $SW_1(\theta) = SW_1^A(\theta)+SW_1^B(\theta)+CTR_{i-1}(\theta)v_{i-i}+CTR_{i}(\theta)v_{i}$.
By assumption, we have $CTR_{i-1}(\theta)v_{i-i}=1$ (as $CTR_{i-1}(\theta)=1$ and $a_{i-1}\neq a_\bot$) and $CTR_{i}(\theta)v_{i}=0$.
We note that $SW_1^A(\theta'')=SW_1^A(\theta)$ and $SW_1^B(\theta'')=SW_1^B(\theta)$.
Furthermore, we note that $CTR_{i-1}(\theta'')v_{i-i}+CTR_{i}(\theta'')v_{i}=1$, as $v_{i-i}=0$ and $CTR_{i}(\theta'')=1$.
So we can conclude that $SW_1(\theta)=SW_1(\theta'')$.
By repeatedly applying the above procedure on $\theta''$ we can obtain an allocation $\theta'$ containing no pair of ads $(a_{i-1},a_i)$ where $\gamma_{i-1,i}=0$ and such that $SW_1(\theta) = SW_1(\theta')$.
  \end{proof}

\begin{lemma}\label{lemma:no_gamma_0_1_equals_no_gamma_0_k}
Let $\theta$ be an allocation  such that no pair of ads $(a_{i-1},a_i)$ exists where $\gamma_{i-1,i}=0$.
Then $SW_1(\theta)=SW_K(\theta)$.
\end{lemma}
\begin{proof}
The claim follows from the fact that $\forall i\in \N$, $CTR_i(\theta) = 1$ for both \B--\caa$(1)$-r and \B--\caa$(K)$-r if $\theta$ does not contain any pair of ads $(a_{i-1},a_i)$ for which $\gamma_{i-1,i}=0$.
  \end{proof}

\begin{proposition}\label{prop:FNE_1-r}
\caa$(1)$-r is APX-hard.
\end{proposition}

\begin{proof}
We prove that the subproblem \B--FNE$_{aa}(1)$-r is APX--hard via an approximation preserving reduction from the APX-hard problem \B--FNE$_{aa}(K)$-r (Theorem \ref{thm:CNFE_inapproximability}). In particular, we will show that computing an approximate solution for \B--FNE$_{aa}(1)$-r is not easier than \B--\caa$(K)$-r on the same instance.

We will first prove that $SW_K(\theta^*_K)\leq SW_1(\theta^*_1)$ holds, where $\theta^*_K$ and $\theta^*_1$ denote, respectively, the optimal allocation for \B--\caa$(K)$-r and \B--\caa$(1)$-r.
For the sake of contradiction, let us suppose that $SW_K(\theta^*_K) > SW_1(\theta^*_1)$.
We can assume without loss of generality that $\theta^*_K$ does not contain a pair $(a_{i-1},a_{i})$ such that $\gamma_{i-1,i}=0$, as replacing $a_{i-1}$ with $a_\bot$ would yield an allocation with a non-decreasing SW value.
By Lemma \ref{lemma:no_gamma_0_1_equals_no_gamma_0_k} and by hypothesis we have that $ SW_1(\theta^*_K)= SW_K(\theta^*_K) >SW_1(\theta^*_1)$, which contradicts the optimality of $\theta^*_1$.

We are now going to prove that given an $\alpha$--approximate solution $\theta_1^\alpha$ to the objective of \B--FNE$_{aa}(1)$-r we can compute in polynomial time an approximate solution $\theta_K^\alpha$ to the objective of \B--FNE$_{aa}(K)$-r such that $SW_1(\theta_1^\alpha)\leq SW_K(\theta_K^\alpha)$. This is easily done by replacing $a_{i-1}$ with $a_\bot$ for each couple of ads $(a_{i-1},a_i)$ in $\theta_1^\alpha$ such that $\gamma_{i-1,i}=0$, thus obtaining $\theta'^\alpha_1$.
By Lemmata \ref{lemma:no_gamma_0} and \ref{lemma:no_gamma_0_1_equals_no_gamma_0_k} we finally conclude that $SW_1(\theta_1^\alpha)=SW_1(\theta'^\alpha_1)=SW_K(\theta'^\alpha_1)$.
  \end{proof}

\begin{proposition}\label{prop:FNE_1-nr}
\caa$(1)$-nr is APX-hard.
\end{proposition}

\begin{proof}
We conduct the proof by reduction from problem \B--\caa$(1)$-r.
In particular, we add to the instance of \B--\caa$(1)$-r $K$ new ads $ \{a_{N+1},\ldots,a_{N+K} \}$ such that: (\emph{i}) $v_{j} =0$ for all $j\in \{N+1, \ldots, N+K\}$ and (\emph{ii}) $\gamma_{i,j}=\gamma_{j,i}=1$ for all $i\in \{1, \ldots, N+K\}$ and $j\in \{N+1, \ldots, N+K\}$.
Let $\theta_{nr}^\alpha$ be an $\alpha$-approximate solution for the so-defined \caa$(1)$-nr problem. 
We can assume w.l.o.g. that $\theta_{nr}^\alpha$ does not contain any $a_\bot$, as in the no-reset model we can always allocate any non-allocated ad to an empty slot obtaining a non-decreasing $SW$ value.
We observe that, from a generic allocation $\theta_{nr}$, it is possible to obtain an allocation $\theta_r$ by substituting any ad $a_j$, $j \in \{N+1, \ldots, N+K\}$, in $\theta_{nr}$ with $a_\bot$ s.t. $SW^r(\theta_r)=SW^{nr}(\theta_{nr})$, and vice versa. Thus, from $\theta_{nr}^\alpha$ we can obtain an allocation $\theta_{r}^\alpha$ s.t. $SW^r(\theta_r^\alpha) = SW^{nr}(\theta_{nr}^\alpha)$; $SW^{x}(\theta)$ denoting the social welfare of $\theta \in \Theta$ in the model with reset $x\in \{r,nr\}$.
Furthermore, let $\theta^*_r$ and $\theta^*_{nr}$ be the optimal solutions, respectively, for \B--\caa$(1)$-r and the \caa$(1)$-nr defined by our reduction.
According to the observations above, it is easy to check that $SW^r(\theta^*_r)=SW^{nr}(\theta^*_{nr})$ holds.
In fact, let $\tilde{\theta}_{nr}$ be the solution obtained from $\theta_r^*$ by substituting each $a_\bot$ with an ad $a_j$, $j \in \{N+1, \ldots, N+K\}$.
Then $SW^r(\theta^*_r)=SW^{nr}(\tilde{\theta}_{nr})$.
Furthermore, $SW^{nr}(\tilde{\theta}_{nr})=SW^{nr}(\theta^*_{nr})$, as otherwise if $SW^{nr}(\tilde{\theta}_{nr})<SW^{nr}(\theta^*_{nr})$ we could translate $\theta^*_{nr}$ into a solution $\tilde{\theta}_r$ for \B--\caa$(1)$-r such that $SW^r(\theta^*_r)<SW^r(\tilde{\theta}_r)$.
A similar argument holds if we consider the allocation $\tilde{\theta}_{r}$ obtained by substituting all ads $a_j$, $j \in \{N+1, \ldots, N+K\}$, in $\theta^*_{nr}$ with $a_\bot$. Finally, $SW^r(\theta_r^\alpha) = SW^{nr}(\theta_{nr}^\alpha) \geq \alpha SW^{nr}(\theta_{nr}^*) = \alpha SW^r(\theta_{r}^*)$.
  \end{proof} 

\section{FNE$_{aa}^+(c)$-nr is APX-complete for constant $\gamma_{min}$}
\begin{theorem}
FNE$_{aa}^{+}$(1)-nr is APX-hard.
\end{theorem}
\begin{proof}
Let $\{\gamma_{min},1\}$-FNE$^+_{aa}(1)$-nr denote the subclass of FNE$_{aa}^{+}$(1)-nr
where $\gamma_{ij}\in\{\gamma_{min},1\}$ for all $i,j\in \N$ and
a given $0<\gamma_{min}<1$. We prove the APX-hardness of FNE$_{aa}^{+}$(1)-nr by an approximation
preserving reduction from problem $\mathcal{B}$-\caa$(1)$-nr (proved APX-hard in Proposition \ref{prop:FNE_1-nr}) to problem $\{\gamma_{min},1\}$-FNE$^+_{aa}(1)$-nr:
we prove the existence of an $\alpha$-approximate algorithm for $\{\gamma_{min},1\}$-FNE$^+_{aa}(1)$-nr to imply
the existence of a $2\alpha$-approximate algorithm for $\mathcal{B}$-\caa$(1)$-nr.

The instance of $\{\gamma_{min},1\}$-FNE$^+_{aa}(1)$-nr is obtained from the
instance of $\mathcal{B}$-\caa$(1)$-nr by simply setting $\gamma'_{i,j}=\gamma_{min}=\frac{1}{K-1}$
for all $i,j\in \N$ such that $\gamma_{i,j}=0$ in the given instance
of $\mathcal{B}$-\caa$(1)$-nr, $\gamma'_{i,j}=1$ otherwise.

Let $\theta_{\gamma_{min}}^{*}$ and $\theta_{\mathcal{B}}^{*}$ be an optimal solution for problems $\{\gamma_{min},1\}$-FNE$^+_{aa}(1)$-nr
and $\mathcal{B}$-\caa$(1)$-nr, respectively. We have $SW(\theta_{\mathcal{B}}^{*})\leq SW(\theta_{\gamma_{min}}^{*})$. Indeed,
if there is no $(a_{i-1},a_{i})\in\theta_{\mathcal{B}}^{*}$
s.t. $\gamma_{i-1,i}=0$ then $SW(\theta_{\mathcal{B}}^{*})=SW(\theta_{\gamma_{min}}^{*})$, whereas if there is a pair $(a_{i-1},a_{i})\in\theta_{\mathcal{B}}^{*}$
s.t. $\gamma_{i-1,i}=0$ then $SW(\theta_{\mathcal{B}}^{*})<SW(\theta_{\gamma_{min}}^{*})$.

Let now $\theta_{\gamma_{min}}$ be an $\alpha$-approximation of
 $\{\gamma_{min},1\}$-FNE$^+_{aa}(1)$-nr and let $\theta_{\mathcal{B}}$
be the corresponding solution for $\mathcal{B}$-\caa$(1)$-nr. (I.e., $\theta_{\mathcal{B}}$ is the solution $\theta_{\gamma_{min}}$ where the $\gamma_{min}$ externalities weigh 0.) We now prove that $SW(\theta_{\gamma_{min}})\leq 2 SW(\theta_{\mathcal{B}})$.
We have $SW(\theta_{\mathcal{B}})=1+\mathcal{P}(\theta_{\mathcal{B}})$,
where $\mathcal{P}(\theta_{\mathcal{B}})\leq K-1$ denotes the number of
pairs $(a_{i-1},a_{i})$ of ads in $\theta_{\mathcal{B}}$ such that
$\gamma_{i-1,i}=1$. Likewise, $SW(\theta_{\gamma_{min}})=1+\mathcal{P}(\theta_{\gamma_{min}})+(K-1-\mathcal{P}(\theta_{\gamma_{min}}))\cdot\gamma_{min}$. 
By construction, $\mathcal{P}(\theta_{\mathcal{B}})=\mathcal{P}(\theta_{\gamma_{min}})=\mathcal{P}$,
from which it follows that $SW(\theta_{\gamma_{min}})\leq 2 \cdot SW(\theta_{_{\mathcal{B}}})$
is equivalent to $1+\frac{K-1-\mathcal{P}}{1+\mathcal{P}}\gamma_{min}\leq 2$. This is proved by noticing that
$1+\frac{K-1-\mathcal{P}}{1+\mathcal{P}}\gamma_{min} \leq 1+\frac{K-1}{1+\mathcal{P}}\gamma_{min}=\frac{\mathcal{P}+2}{\mathcal{P}+1}$, where last equality follows from definition of $\gamma_{min}$.
  \end{proof}

\subsection{Approximation algorithm}
We now prove that any $\alpha$-approximate algorithm for Weighted 3-Set Packing (W3SP) can be turned into an $(\alpha\gamma_{min}^c)$--approximation algorithm for FNE$^+_{aa}(c)$--nr.

Given a universe $U$ and a collection of its subsets each of cardinality at most 3 and associated to a weight, W3SP consists of finding a sub-collection of pairwise-disjoint subsets of maximal weight. Several constant-ratio approximate algorithms are known in literature to solve this problem, e.g., the algorithm in \cite{Berman00} provides a $1/2$-approximation. 
We now present a reduction from FNE$^+_{aa}(c)$-nr to W3SP, similar in spirit to that defined, for positive only externalities, in \cite{Fotakis}. \begin{theorem}\label{th:apx_3-set-pack}
Given an $\alpha$--approximate algorithm for problem W3SP, we can obtain an $(\alpha \gamma_{min}^c)$-approximation algorithm for problem FNE$^+_{aa}(c)$-nr.
\end{theorem}
\begin{proof}
Given an instance of FNE$^+_{aa}(c)$-nr, we obtain an instance of W3SP by means of the following reduction.
To simplify the presentation, we suppose that $K$ is even (the proof can be easily extended for an odd $K$).
We divide $K$ into $K/2$ blocks of two slots each.
We construct a collection of $\frac K 2\cdot \binom{N}{2}$ sets, each set having the form $\{a_i, a_j, p\}$, where $p\in\{1,3,5,\ldots,K-1\}$ and $i, j \in \N$. The weight of a set is defined as the maximum social welfare that ads $a_i$ and $a_j$ can provide when assigned to slots $s_p$ and $s_{p+1}$ without taking into considerations the externalities of $a_i$ and $a_j$ on the ads allocated to slots $s_m$, $m \neq p, p+1$. Specifically, 
$ W(a_i, a_j, p) = \max\{ \Lambda_p  q_i v_i + \Lambda_{p+1} \gamma_{i,j} q_j v_j, \Lambda_p   q_j v_j + \Lambda_{p+1} \gamma_{j,i} q_i v_i \}.
$
Note that there is an immediate mapping between solutions of W3SP and FNE$^+_{aa}(c)$-nr. For a solution $\theta_S$ of W3SP, let $W(\theta_S)$ denote its total weight. 
Now, let $\theta_S^*$ and $\theta^*$ denote, respectively, an optimal allocation for W3SP and an optimal allocation for FNE$^+_{aa}(c)$-nr. Furthermore, let $\theta_S^{\alpha}$ be an $\alpha$-approximate solution for W3SP, and $\theta^{\alpha}$ be the corresponding solution to FNE$^+_{aa}(c)$-nr. Since in W3SP, outer-block externalities are not taken into consideration, we have: 
$W(\theta^*_S) \geq SW(\theta^*)$ and  
$SW(\theta^\alpha) \geq \gamma_{\min}^c W(\theta_S^\alpha)$. 
From these inequalities we obtain: $SW(\theta^{\alpha}) \geq \gamma_{\min}^c W(\theta_S^{\alpha}) \geq \alpha \gamma_{\min}^c W(\theta^*_S) \geq \alpha \gamma_{min}^c SW(\theta^*)$.
  \end{proof}
\begin{corollary}\label{corl:constant_apx}
If $\gamma_{min}$ is bounded from below by a constant (i.e., $\gamma_{min}\in \Omega(1)$), then FNE$^+_{aa}(c)$-nr is approximable within a constant factor.
\end{corollary}
It can be easily shown that the above algorithm is not monotone.

\begin{theorem}
The algorithm of Theorem \ref{th:apx_3-set-pack} is not monotone
\end{theorem}
\begin{proof}
Consider an instance $I$ of FNE$^+_{aa}(1)$-nr with $N=K=4$ wherein $\Lambda_3 \gamma_{z,4} < \Lambda_4 \gamma_{3,4}$, for $z \in \{1,2\}$, $v_1, v_2 \gg v_3, v_4$ and $\gamma_{1,2}=\gamma_{2,1}=1$ so that $W(a_1, a_2, 1)$ is much bigger than any other $W(a_i, a_j, 1)$. Therefore, any reasonable approximation of the W3SP instance constructed upon $I$ must return sets $\{a_1, a_2, 1\}$ and $\{a_3, a_4, 3\}$. Additionally consider $v_4 < \frac{\Lambda_4 \gamma_{4,3}}{\Lambda_3^2-\Lambda_3\Lambda_4 \gamma_{3,4}}$ so that $W(a_3, a_4, 3)=\Lambda_3 q_3 v_3 + \Lambda_{4} \gamma_{3,4} q_4 v_4$. So the solution $\theta$ returned by the algorithm run on $I$ places $a_4$ in $s_4$, resulting in $CTR_4(\theta)=q_4 \Lambda_4 \gamma_{3,4}$. Take now the instance $I'$ defined as $I$ except that $v_1, v_2 \gg v_4' > \frac{\Lambda_4 \gamma_{4,3}}{\Lambda_3^2-\Lambda_3\Lambda_4 \gamma_{3,4}} > v_4$. As before, the approximation algorithm for W3SP will return sets $\{a_1, a_2, 1\}$ and $\{a_3, a_4, 3\}$ but this time $W'(a_3, a_4, 3)=\Lambda_3 q_4 v_4 + \Lambda_{4} \gamma_{4,3} q_3 v_3$. Therefore, the solution $\theta'$ returned by the algorithm run on $I'$ places ad $a_4$ in slot $s_3$, i.e., $CTR_4(\theta')=q_4 \Lambda_3 \gamma_{z,4}$, where $z \in \{1,2\}$ is the ad placed in slot $s_2$ in the allocation $\theta'$. The algorithm is therefore not monotone and cannot be used to design a truthful mechanism.  
\end{proof}

\section{Approximating \caa$(c)$-nr}
Similarly to the case $c=K$, Color Coding can be applied to design an optimal exponential-time algorithm finding the optimal solution and a simple modification of such algorithm returns a $\frac{\log(N)}{2\min\{N,K\}}$ approximation in polynomial time. While the basic idea is the same, some details change here. 

We denote by $S\subseteq C$ a subset of colors and by $\delta(a)$ a function returning the color assigned to $a$.
Given a coloring $\delta$, the best colorful allocation
 is found by dynamic programming.
For $|S|>c$, $W(S,\langle a_{h_0},\ldots,a_{h_{c}} \rangle)$ contains the value of the best allocation with colors in $S$ in which the last $c+1$ ads are $a_{h_0},\ldots,a_{h_{c}}$ from top to bottom. (The definition naturally extends for $|S| \leq c$.) Starting from $W(\emptyset, \langle \rangle)=0$,
we can compute $W$ recursively. For instance, for $|S|>c$, $W(S\cup\{\delta(a_{h_c})\},\langle a_{h_0}, \ldots,a_{h_{c}}  \rangle) = \Lambda_{|S|+1}v_{h_c}q_{h_c}\prod_{i=0}^{c-1}\gamma_{h_{i},h_{i+1}} +\max_{a}W(S,\langle a, a_{h_0},\ldots,$ $a_{h_{c-1}} \rangle)$ if $\delta(a_{h_c}) \not\in S$ and 
$-\infty$ otherwise. Given a random coloring, the probability that the ads composing the best allocation are colorful is $\frac{K!}{K^K}$. Thus, repeating the procedure $r{e^K}$ times, where $r\geq 1$, the probability of finding the best allocation is $1-e^{-r}$.
The complexity is $O((2e)^KKN^{c+2})$.
The algorithm can be derandomized with an additional cost of $O(\log^2(N))$.

By applying the above algorithm to the first $K'$ slots, $K' = \min\{K,\lceil \log(N) \rceil \}$, we obtain an algorithm with complexity $O(K^{3.5}N^{c+2}\log_2^2(N))$. We observe that if $c$ is not a constant, the complexity is exponential. It is not too hard to note that such an algorithm is $\frac{\log(N)}{2\min\{N,K\}}$-approximate. Moreover, this algorithm is MIR and as such can be used to design a truthful mechanism.

\section{Conclusions}
We enrich the literature on externalities in SSAs by introducing more general ways to model slot- and ad-dependent externalities, while giving a (nearly) complete picture of the computational complexity of the problem. 
In detail, we enrich the naive model of SSAs by adding: (\emph{i}) the concepts of limited user memory  
(\emph{ii}) contextual externalities and (\emph{iii})
refreshable user memory (i.e., reset  model).

This gives rise to the \csa{} model, where ad- and slot-dependent externalities are factorized as in the cascade model and the \caa{} model, where the externalities and not factorized.

We satisfactorily solve the problem for \csa{}, whereas our results leave unanswered a number of interesting questions, with regards to both approximation and truthfulness for \caa{}. The parameter $c$ is central to this list. 
If $c$ is constant, then we do not know whether a constant approximation algorithm for \caa$(c)$ exists; this holds also for the special case of FNE$^+_{aa}(c)$-nr when $\gamma_{min}$ is not a constant. In the latter case, when $\gamma_{\min}$ is instead constant we are not aware of any truthful constant approximation mechanism. 
Motivated by the fact that \caa{}-r is, apparently, an easier problem than \caa{}-nr, we believe that an interesting direction for future research is to study reset in more detail in order to understand its role w.r.t. the relatively harder \caa{}-nr.

\bibliography{citations}
\bibliographystyle{abbrv}


\end{document}